\documentclass[12pt,letterpaper]{article}
\usepackage[utf8]{inputenc}
\usepackage[margin=1in]{geometry}

\usepackage{graphicx}
\usepackage{subcaption}

\usepackage{mathtools}
\usepackage{amsthm}

\DeclareMathOperator*{\argmin}{arg\,min}

\newtheorem{theorem}{Theorem}
\newtheorem{lemma}{Lemma}

\theoremstyle{definition}
\newtheorem{example}{Example}

\usepackage{enumitem}

\usepackage{algorithm,algcompatible}

\usepackage[onehalfspacing]{setspace}

\usepackage[round]{natbib}

\newcommand{\primal}{\textsc{P}}
\newcommand{\dual}{\textsc{D}}
\newcommand{\alg}{\textsc{A}}

\newcommand{\type}{\tau}

\newcommand{\chosen}{\textrm{\em matched}}
\newcommand{\notchosen}{\textrm{\em unmatched}}
\newcommand{\nonadapt}{\textrm{\em unknown}}

\newcommand{\kmin}{k_{\min}}
\newcommand{\kmax}{k_{\max}}

\title{Understanding Zadimoghaddam's Edge-weighted Online Matching Algorithm: Unweighted Case}

\author{
Zhiyi Huang%
\thanks{The University of Hong Kong. Email: zhiyi@cs.hku.hk}
\and
Runzhou Tao%
\thanks{Columbia University. This work was done while the author was at IIIS, Tsinghua University. Email: runzhou.tao@columbia.edu}
}

\date{Last Update: October 2019}

\begin{document}

\maketitle

\begin{abstract}
    This article identifies a key algorithmic ingredient in the edge-weighted online matching algorithm by \cite{Zadimoghaddam/arXiv/2017} and presents a simplified algorithm and its analysis to demonstrate how it works in the unweighted case.
\end{abstract}

\section{Introduction}
\label{sec:intro}

Online edge-weighted bipartite matching is a major open problem in the area of online algorithms. 
Consider a bipartite graph $G = (L, R, E)$ where $L$ and $R$ denote the sets of left-hand-side (LHS) and right-hand-side (RHS) vertices and $E \subseteq L \times R$ denotes the set of edges.
Further, the graph is edge-weighted;
every edge $e \in E$ is associated with a nonnegative weight $w(e) \ge 0$.
The LHS is given upfront, while vertices on on the RHS arrive online one at a time. 
On the arrival of an online vertex $j \in L$, the algorithm observes the set of edges adjacent to $j$ and, thus, the set of offline neighbors of $j$, denoted as $N(j)$.
The algorithm needs to immediately decide whether to match $j$ to an offline neighbor, and if so which one.
The algorithm is allow to match multiple online vertices and, thus, the corresponding edges to the same offline vertex, but only the one with the largest edge weight counts.
This is equivalent to assuming that the algorithm can rematch an offline vertex to a new online vertex, disposing its previously matched online neighbor. 
Hence, it is also referred to as the \emph{free disposal} model.
The goal is to maximize the total weight of the matching.

This problem generalizes the unweighted online bipartite matching problem by \citet{KarpVV/STOC/1990} and the vertex-weighted problem by \citet{AggarwalGKM/SODA/2011}, both of which admit online algorithms with the optimal $1 - \frac{1}{e}$ competitive ratio.
The $1 - \frac{1}{e}$ ratio also serves as an upper bound for the competitive ratio of any online algorithm for the edge-weighted problem considered in this article.

On the other hand, it is folklore that a simple greedy algorithm, which matches each online vertex to the offline neighbor that provides the maximum marginal gain (e.g., the weight of the new edge minus that of the previous one), is $\frac{1}{2}$-competitive.

Pinning down the optimal competitive ratio for the edge-weighted online bipartite matching problem between $\frac{1}{2}$ and $1 - \frac{1}{e}$ has been an open problem since at least \citet{FeldmanKMMP/WINE/2009}. 
They study a variant that counts the heaviest $n > 1$ edges matched to each offline vertex; 
it can be viewed as a fractional version of the problem.


\cite{Zadimoghaddam/arXiv/2017} introduces an algorithm and gives a competitive ratio that beats the $\frac{1}{2}$ barrier; 
however, it is difficult to fully understand the algorithm and its analysis even for experts.
We seek to provide an alternative and more accessible exposition of the algorithm and analysis by \cite{Zadimoghaddam/arXiv/2017}, fixing any minor bugs therein.

To this end, this article identifies one technical ingredient in the original paper that we believe to be the key of breaking the $\frac{1}{2}$ barrier.
This ingredient is the only thing that we keep from the original paper by \cite{Zadimoghaddam/arXiv/2017}, with minor changes;
we take the liberty to change everything else wherever we see fit to improve the exposition.
This article is devoted to introducing this technical ingredient and explaining how it works in the special case of unweighted graphs, i.g., all edges have unit weight.
Generalization to the weighted case will be covered in a follow-up article.

Finally, due to the nature of this article, we will generally prioritize for simplicity of the algorithm and analysis rather than optimality of the competitive ratio.



\section{Online Primal Dual}
\label{sec:online-primal-dual}

This section presents a brief introduction of the online primal dual framework for analyzing the competitive ratios of online matching algorithms. 
It is widely used in the literature of online algorithms in general, and is developed for online matching in a series of papers by \cite{BuchbinderJN/ESA/2007}, \cite{DevanurJ/STOC/2012}, \cite{DevanurJK/SODA/2013}, \cite{WangW/ICALP/2015}, \cite{DevanurHKMY/TEAC/2016}, \cite{HuangKTWZZ/STOC/2018}, \cite{HuangTWZ/ICALP/2018}, \cite{HuangPTTWZ/SODA/2019}.
See \cite{DevanurJK/SODA/2013} for an in-depth discussion on this topic.

Recall that this article focuses on the special case of unweighted graphs.
Consider the following standard matching linear program (LP):
\begin{align*}
    \textrm{maximize} \quad & \sum_{(i,j) \in E} x_{ij} \\
    \textrm{subject to} \quad & \sum_{j \in N(i)} x_{ij} \le 1 && \forall i \in L \\
    & \sum_{i \in N(j)} x_{ij} \le 1 && \forall j \in R \\[.5ex]
    & x_{ij} \ge 0 && \forall (i,j) \in E
\end{align*}

In the offline problem, $x_{ij}$ is the indicator of whether edge $(i,j)$ is in the matching.
Recall that we allow an offline vertices to be matched multiple times in the online problem, but each offline vertex contributes at most $1$ in the objective.
Further, the online algorithm may be randomized.
To this end, we may informally interpret $x_{ij}$ as the probability that edge $(i,j)$ is the first edge matched to vertex $i$ by the online algorithm.

Next, consider the following dual LP:
\begin{align*}
    \textrm{minimize} \quad & \sum_{i \in L} \alpha_i + \sum_{j \in R} \beta_j \\
    \textrm{subject to} \quad & \alpha_i + \beta_j \ge 1 && \forall (i,j) \in E \\[1ex]
    & \alpha_i \ge 0 && \forall i \in L \\[1ex]
    & \beta_j \ge 0 && \forall j \in R
\end{align*}

Let $\primal$ and $\dual$ denote the values of the primal and dual objectives respectively.
Let $\alg$ denote the expected objective of the algorithm.
An online primal dual analysis maintains a pair of nonnegative primal and dual assignments at all time satisfying a set of conditions formulated in the next lemma.

\begin{lemma}
    \label{lem:online-primal-dual}
    Suppose for some $0 < \Gamma \le 1$, the following conditions hold at all time:
    \begin{enumerate}
        \item (Objectives) $\alg \ge \primal \ge \dual$.
        \item (Approximate Dual Feasibility) For any edge $(i, j) \in E$ that has arrived, we have:
            \[
                \alpha_i + \beta_j \ge \Gamma
                ~.
            \]
    \end{enumerate}
    Then, the algorithm is $\Gamma$-competitive.
\end{lemma}

\begin{proof}
    Since the conditions hold at all time, they hold at the end of the algorithm. 
    By the second condition, dividing the dual assignment by $\Gamma$ gives a feasible dual.
    Hence, $\Gamma^{-1} \cdot \dual$ is the objective of a feasible dual.
    By weak duality of LP, it is an upper bound of the optimal.
    Putting together with the first condition proves the lemma.
\end{proof}

%
A canonical online primal dual analysis usually sets the primal faithfully according to the matching maintained by the online algorithm:
let $x_{ij}$ be the probability that edge $(i,j)$ is the first edge matched to vertex $j$.
Then, we have the first inequality in the objective condition holds with equality, i.e., $\alg = \primal$.
Further, the typical way of setting the dual can be viewed as a gain splitting process.
Whenever an edge $(i,j)$ is chosen by the online algorithm and it is the first edge matched to the offline vertex $i$, the primal objective increases by $1$ and, thus, a total gain of $1$ is split between $\alpha_i$ and $\beta_j$.
That is, $\alpha_i$ and $\beta_j$ are increase by their respectively amounts which sum to $1$.
If the algorithm is randomized, the dual variables will be the expected share that the vertices get in the gain splitting process.
The gain splitting process needs to be tailored for the problem and algorithm at hand to satisfy the approximate dual feasibility condition.
Then, the increment in the primal objective always matches that in the dual objective and, thus, the second inequality in the objective condition also holds with equality, i.e., $\primal = \dual$.

%
\section{Thought Experiments}
\label{sec:thought-experiment}

This section presents two failing algorithms as thought experiments in order to build up the intuition that leads to the final algorithm.
The online primal dual analysis in this section follows the aforementioned canonical format.

\subsection{Two-choice Greedy}
\label{sec:nonadaptive-greedy}

Let us start with a simple greedy algorithm that always matches to the neighbor(s) with the largest marginal gain, which equals the probability that the neighbor is still unmatched, over the randomness in previous rounds.
To be consistent with the subsequent discussions, we consider a greedy algorithm that, in the presence of at least two neighbors both with the largest marginal gain, picks two of them as candidates, say, lexicographically, and matches to them with equal probabilities.
We say that such a round that manages to find two candidates is \emph{randomized}; 
a round that finds only one candidate is \emph{deterministic}.
If all neighbors have already been matched with certainty, it is an \emph{unmatched} round.

On the one hand, the marginal gain of a vertex $i \in L$ becomes $0$ once it gets involved in a deterministic round and, thus, is no longer relevant in future rounds.
On the other hand, the marginal gain of a vertex $i \in L$ decreases by half every time it gets involved in a randomized round.
Hence, the neighbors with the largest marginal gain are those that are never involved in any deterministic round, and are involved in the fewest number of randomized rounds.

See Algorithm~\ref{alg:two-choice-greedy} for a formal definition.

\begin{algorithm}[t]
    \caption{Two-choice Greedy}
    \label{alg:two-choice-greedy}
    \begin{algorithmic}
        \medskip
        \STATEx \textbf{State variables:}
        \begin{itemize}
            \item $k_i \ge 0$, number of randomized rounds in which an offline vertex $i$ has been chosen;
            $k_i = \infty$ if it has been chosen in a deterministic round.
        \end{itemize}
        \smallskip
        \STATEx \textbf{On the arrival of an online vertex $j \in R$:}
        \begin{enumerate}
            \item Let $B(j) = \argmin_{i \in N(j)} k_i$ be the set of neighbors of $j$ with the smallest \emph{finite} $k_i$.
            \item If $|B(j)| \ge 2$, pick two neighbors $i_1, i_2 \in B(j)$, say, lexicographically;
            match $j$ to $i_1$ or $i_2$ each with probability $\frac{1}{2}$.
            \hspace*{\fill}
            \textbf{(randomized round)}
            \item If $|B(j)| = 1$, let $i \in B(j)$;
            match $j$ to $i$.
            \hspace*{\fill}
            \textbf{(deterministic round)}
            \item If $|B(j)| = 0$, leave $j$ unmatched.
            \hspace*{\fill}
            \textbf{(unmatched round)}
            \item Update $k_i$'s accordingly.
        \end{enumerate}
    \end{algorithmic}
    \smallskip
\end{algorithm}

\begin{theorem}
    \label{thm:greedy}
    Two-choice greedy is $\frac{1}{2}$-competitive and it is tight for the algorithm.
\end{theorem}

\subsubsection{Competitive Ratio: Lower Bound}

Experts may skip this subsection.
It merely serves as a warm-up case for other readers to get familiar with the online primal dual framework.

Consider the following online primal dual analysis.
At all time, let $x_{ij}$ be the probability that edge $(i,j)$ is the first edge matched to vertex $i$ by the algorithm. 
Let $x_i = \sum_{j \in N(i)} x_{ij}$ be the probability that an offline vertex $i$ is matched. 
For each online vertex $j \in R$, update the primal and dual variables as follows by splitting the gain equally between the two endpoints for every edge ($x_{ij}$'s are $0$ by default unless stated otherwise).
\begin{itemize}
    \item \textbf{Randomized rounds:~} Suppose it is a randomized round, with $\kmin = \min_{i \in N(j)} k_i$ at the time of the match. 
    Suppose $i_1$ and $i_2$ are the chosen neighbors.
    Then, let $x_{ij}$ be $2^{-\kmin-1}$ and, thus, increase $x_i$ by the same amount, for $i = i_1, i_2$.
    Increase $\alpha_i$ by $2^{-\kmin-2}$ for $i = i_1, i_2$, and let $\beta_j = 2^{-\kmin-1}$.
    \item \textbf{Deterministic rounds:~} Suppose it is a deterministic round, with $\kmin = \min_{i \in N(j)} k_i$ at the time of the match, and $i$ is the chosen neighbor.
    Then, let $x_{ij}$ be $2^{-\kmin}$ and, thus, increase $x_i$ by the same amount (after which it equals $1$). 
    Increase $\alpha_i$ by $2^{-\kmin-1}$ (after which it equals $\frac{1}{2}$), and let $\beta_j = 2^{-\kmin-1}$.
    \item \textbf{Unmatched rounds:~} Suppose it is an unmatched round. 
    Then, let $\beta_j = 0$.
\end{itemize}

The next few lemmas follow straightforwardly by how the variables are updated above.

\begin{lemma}
    For any offline vertex $i$, we have at all time:
    \[
        x_i = 1 - 2^{-k_i}
        ~.
    \]
\end{lemma}

\begin{lemma}
    \label{lem:greedy-alpha}
    For any offline vertex $i$, we have at all time:
    \[
        \alpha_i = \frac{1 - 2^{-k_i}}{2}
        ~.
    \]
\end{lemma}

\begin{lemma}
    The primal and dual objectives are equal at all time.
\end{lemma}

It remains to analyze approximate dual feasibility, as stated in the next lemma.

\begin{lemma}
    For any edge $(i,j) \in E$, we have the following at the end of the algorithm:
    \[
        \alpha_i + \beta_j \ge \frac{1}{2}
        ~.
    \]
\end{lemma}

\begin{proof}
    Consider the moment when vertex $j$ arrives and the value of $k_i$ then.
    By Lemma~\ref{lem:greedy-alpha}, we have:
    \[
        \alpha_i = \frac{1 - 2^{k_i}}{2}
        ~.
    \]
    
    Further, let $\kmin = \argmin_{i' \in N(j)} k_{i'}$ at the time.
    Note that $\kmin \le k_i$. 
    If it is a randomized round or a deterministic round, we have:
    \[
        \beta_j = 2^{-\kmin-1} \ge 2^{-k_i-1}
        ~.
    \]
    
    Thus, together with the above bound of $\alpha_i$, we get the inequality stated in the lemma.
    
    If it is an unmatched round, we have $k_i = \kmin = \infty$, and $\beta_i = 0$.
    The contribution from $\alpha_i$ alone satisfies the inequality stated in the lemma.
\end{proof}

\subsubsection{Competitive Ratio: Upper Bound}

For ease of presentation, we consider a version of the algorithm which picks neighbors in the reverse lexicographical order in randomized rounds.
Consider the following instance.

\begin{example}
    Consider a bipartite graph with $n = 3^k$ vertices on both sides for some large integer $k$.
    The first $\frac{n}{3} = 3^{k-1}$ online vertices are connected to all offline vertices.
    Then, the first one third of the remaining online vertices, i.e., $\frac{1}{3} \cdot \frac{2n}{3} = 2 \cdot 3^{k-2}$ of them in total, are connected to the last $\frac{2n}{3} = 2 \cdot 3^{k-1}$ offline vertices.
    In general, for any $0 \le i < k$, the first one third of the last $\big( \frac{2}{3} \big)^i n = 2^i \cdot 3^{k-i}$ vertices, i.e., $2^i \cdot 3^{k-i-1}$ of them in total, are connected to the last $\big( \frac{2}{3} \big)^i n = 2^i \cdot 3^{k-i}$ offline vertices.
    Finally, let there be a perfect matching between the last $2^k$ offline vertices and the last $2^k$ online vertices.

\end{example}

First, note that there is a perfect matching, with the $i$-the online vertex matching to the $i$-th offline vertex.
Hence, the optimal is $n$.

Next, consider the performance of the online algorithm.
The first $\frac{n}{3} = 3^{k-1}$ vertices are connected to all offline vertices.
They are matched to the last $\frac{2}{3}$ fraction of the offline vertices in randomized rounds.
That is, their correct neighbors in the perfect matching are left unmatched, while other offline vertices are matched by half.
Then, the first one third of the remaining online vertices, i.e., $\frac{1}{3} \cdot \frac{2n}{3} = 2 \cdot 3^{k-2}$ of them in total, are matched to the last $(\frac{2}{3})^2$ fraction of the offline vertices in randomized rounds.
That is, their correct neighbors in the perfect matching are left matched by only half, while the correct neighbors of subsequent online vertices are now matched by three quarters.
The argument goes on recursively.

Therefore, omitting a lower order term due to the last $2^k = n^{\log_3 2}$ vertices on both sides, the expected size of the matching is:
\begin{align*}
    \bigg( 1 \cdot \frac{1}{3} + \frac{1}{2} \cdot \frac{2}{9} + \cdots + \bigg(\frac{1}{2}\bigg)^k \cdot \frac{2^k}{3^{k+1}} + \cdots \bigg) n 
    & 
    = \bigg( \frac{1}{3} + \frac{1}{9} + \cdots + \frac{1}{3^{k+1}} + \cdots \bigg) n \\[1ex]
    & 
    = \frac{n}{2}
    ~.
\end{align*}

\subsection{Greedy with Perfect Negative Correlation}

Next, we consider an imaginary version of the two-choice greedy algorithm that is almost identical to the original version, except that the random bits associated with the same offline vertex in different randomized rounds have perfect negative correlation:
if an offline vertex $i$ is not matched the first time it is chosen in a randomized round, it will be matched the second time. 

It is in general impossible to achieve such perfect negative correlations in the online setting. 
See Section~\ref{sec:perfect-correlation-infeasible} for a brief discussion.
Nevertheless, this subsection presents an online primal dual analysis of this algorithm assuming its feasibility as a thought experiment, to demonstrate that negative correlations lead to a competitive ratio better than $\frac{1}{2}$.

\begin{theorem}
    \label{thm:perfect-correlation}
    Two-choice greedy with perfect negative correlation, if feasible, is $\frac{5}{9}$-competitive and it is tight for the algorithm.
\end{theorem}

\subsubsection{Competitive Ratio: Lower Bound}

Let $\Gamma = \frac{5}{9}$.
We will show that greedy with perfect negative correlation is $\Gamma$-competitive via an online primal dual argument as follows.
Let $x_{ij}$ be the probability that edge $(i,j)$ is the first edge matched to vertex $i$ in the algorithm. 
Let $x_i = \sum_{j \in N(i)} x_{ij}$ be the probability that an offline vertex $i$ is matched. 
For each online vertex $j \in R$, update the primal and dual variables as follows.
\begin{itemize}
    \item \textbf{Randomized rounds:~} Suppose it is a randomized round, with $\kmin = \min_{i \in N(j)} k_i$ at the time of the match, and $i_1$ and $i_2$ are the chosen neighbors.
    Then, let $x_{ij}$ be $\frac{1}{2}$ and, thus, increase $x_i$ by the same amount, for $i = i_1, i_2$.
    If $\kmin = 0$, i.e., it is the first time $i_1$ and $i_2$ are chosen in randomized rounds, increase $\alpha_i$ by $\frac{1-\Gamma}{2}$ for $i = i_1, i_2$, and let $\beta_j = \Gamma$.
    If $\kmin = 1$, i.e., it is the second time $i_1$ and $i_2$ are chosen in randomized rounds, increase $\alpha_i$ by $\frac{3(1 - \Gamma)}{4}$ for $i = i_1, i_2$, and let $\beta_j = \frac{3 \Gamma - 1}{2}$.
    \item \textbf{Deterministic rounds:~} Suppose it is a deterministic round, with $\kmin = \min_{i \in N(j)} k_i$ at the time of the match, and $i$ is the chosen neighbor.
    Then, let $x_{ij}$ be $2^{-\kmin}$ and, thus increase $x_i$ by the same amount (after which it equals $1$). 
    If $\kmin = 0$, i.e., vertex $i$ has never been chosen in randomized rounds, increase $\alpha_i$ by $\Gamma$, and let $\beta_j = 1 - \Gamma$.
    If $\kmin = 1$, i.e., vertex $i$ was chosen in a randomized round before $j$'s arrival, increase $\alpha_i$ by $\frac{3(1-\Gamma)}{4}$, and let $\beta_j = \frac{3\Gamma-1}{4}$.
    \item \textbf{Unmatched rounds:~} Suppose it is an unmatched round. 
    Then, let $\beta_j = 0$.
\end{itemize}

Note that for any offline vertex $i$, the value of $x_i$ is either $0$, or $\frac{1}{2}$, or $1$.
Concretely, $x_i = 0$ if it has never been chosen in any randomized or deterministic rounds;
$x_i = \frac{1}{2}$ if it is chosen in exactly one randomized round but not in any deterministic rounds;
$x_i = 1$ if either it is chosen in two randomized rounds, or it is chosen in a deterministic round.

We first present the lower bounds of $\alpha_i$ for different values of $x_i$.

\begin{lemma}
    For any offline vertex $i$:
    \[
        \alpha_i = 
        \begin{cases}
            0 & \text{\rm if $x_i = 0$;} \\
            \frac{1-\Gamma}{2} & \text{\rm if $x_i = \frac{1}{2}$;} \\
            \Gamma & \text{\rm if $x_i = 1$.}
        \end{cases}
    \]
\end{lemma}

\begin{proof}
    The first two cases follow by definition.
    The last case also follows by definition if $x_i$ directly changes from $0$ to $1$ in a deterministic round.
    If $x_i$ first increases from $0$ to $\frac{1}{2}$ in a randomized round and then further increases to $1$ either in a deterministic round or a randomized round, the value of $\alpha_i$ is:
    \[
        \frac{1-\Gamma}{2} + \frac{3(1-\Gamma)}{4} = \frac{5(1-\Gamma)}{4} = \Gamma
        ~.
    \]
    where the second equality follows by $\Gamma = \frac{5}{9}$.
\end{proof}

The next lemma follows by the definition of the online primal and dual updates.

\begin{lemma}
    The primal and dual objectives are equal at all time.
\end{lemma}

It remains to analyze approximate dual feasibility, as stated in the next lemma.

\begin{lemma}
    For any edge $(i,j) \in E$, we have the following at the end of the algorithm:
    \[
        \alpha_i + \beta_j \ge \Gamma
        ~.
    \]
\end{lemma}

\begin{proof}
    Consider the moment when vertex $j$ arrives and the value of $k_i$ then.
    
    \paragraph{Case 1: $k_i = 0$.}
    In this case, we have $x_i = 0$ and $\alpha_i = 0$ at the time when $j$ arrives.
    If it is a randomized round, we get that $\beta_j = \Gamma$ and, thus, the inequality follows.
    If it is a deterministic round, it must be the case that $j$ is matched to $i$ and, thus, $\alpha_i + \beta_j = 1$.
    
    \paragraph{Case 2: $k_i = 1$.}
    In this case, we have $x_i = \frac{1}{2}$ and $\alpha_i = \frac{1 - \Gamma}{2} = \frac{2}{9}$ at the time when $j$ arrives.
    If it is a randomized round, we have $\beta_j = \Gamma$ if $\kmin = 0$, and $\frac{3\Gamma-1}{2}$ if $\kmin = 1$.
    If it is a deterministic round with $\kmin = 0$, we have $\beta_j = 1 - \Gamma$.
    Therefore:
    \[
        \beta_j \ge \min \bigg\{ \Gamma, \frac{3\Gamma-1}{2}, 1 - \Gamma \bigg\} = \min \bigg\{ \frac{5}{9}, \frac{1}{3}, \frac{4}{9} \bigg\} = \frac{1}{3}
        ~.
    \]
    
    Hence, we get that:
    \[
        \alpha_i + \beta_j \ge \frac{2}{9} + \frac{1}{3} = \frac{5}{9} = \Gamma
        ~.
    \]

    If it is a deterministic round with $\kmin = 1$, on the other hand, it must be the case that $j$ is matched to $i$.
    As a result, on top of having $\alpha_i = \frac{2}{9}$ before the arrival of $j$, $\alpha_i$ and $\beta_j$ further split the gain of $\frac{1}{2}$ due to $j$.
    In total, we have:
    \[
        \alpha_i + \beta_j = \frac{2}{9} + \frac{1}{2} > \Gamma
        ~.
    \]
    
    \paragraph{Case 3: $k_i = 2$ or $\infty$.}
    In this case, we have $\alpha_i = \Gamma$ and, thus, the inequality follows.
\end{proof}

\subsubsection{Competitive Ratio: Upper Bound}

Consider the same instance for the vanilla version, but keeping only the first two rounds.
Concretely, let there be a bipartite graph with $9$ vertices on each side, denoted as $i_1, i_2, \dots, i_9$ and $j_1, j_2, \dots, j_9$, and a perfect matching with $i_k$ matched to $j_k$ for $k = 1, 2, \dots, 9$.
The first three online vertices, $j_1$, $j_2$, and $j_3$, are connected to all offline vertices.
After their arrivals, $i_1$, $i_2$, and $i_3$ are unmatched while the remaining $6$ offline vertices are matched by half.
Then, the next two online vertices, $j_4$ and $j_5$, are connected to the last $6$ offline vertices, i.e., $j_4$ to $j_9$.
After their arrival, $i_4$ and $i_5$ remain matched by half, while $i_6$ to $i_9$ are fully matched. 
Therefore, the algorithm finds a matching of size $\frac{1}{2} \cdot 2 + 1 \cdot 4 = 5$  in expectation, while the optimal matching has size $9$. 
The ratio is $\frac{5}{9}$, matching the lower bound that we show.

\subsubsection{Infeasibility}
\label{sec:perfect-correlation-infeasible}

This subsection presents an example demonstrating that two-choice greedy with perfect negative correlation is infeasible in the online setting.
Consider a graph with $4$ offline vertices, denoted as $1$ to $4$.
The first online vertex, denoted as $5$, is connected to $1$ and $2$.
The second online vertex, denoted as $6$, is connected to $3$ and $4$.
The third online vertex, denoted as $7$, has two possibilities:
it is connected with either $1$ and $3$, or $1$ and $4$.
In the former case, 
the following pairs of edges have perfect negative correlations: $(1,5)$ and $(2,5)$, $(1,7)$ and $(3,7)$, $(1,5)$ and $(1,7)$, and $(3,6)$ and $(3,7)$.
The first two pairs are due to having the same online vertex;
the last two pairs are due to having the same offline vertex.
Hence, we can deduce that $(2,5)$ and $(3,6)$ have perfect positive correlation. 
In the latter case, however, a similar argument gives that $(2,5)$ and $(3,6)$ have perfect negative correlation.
An online algorithm cannot handle both cases simultaneously since the correlation between $(2,5)$ and $(3,6)$ are determined before the arrival of vertex $7$ in the online setting.

\section{Greedy with Partial Negative Correlation}
\label{sec:partial-correlation}

This section presents the actual algorithm and its analysis.
For any $\ell \in \{1, 2\}$, let $-\ell$ denote the other element in $\{1, 2\}$, i.e., $3 - \ell$.

\subsection{Online Correlated Selection}

The main ingredient is a subroutine which we will refer to as the \emph{online correlated selection (OCS)}.
See Algorithm~\ref{alg:OCS}.

Informally, the OCS ensures that (1) the marginal distribution of any particular round is uniform over the candidates, and (2) for any fixed offline vertex, the randomness in the rounds in which it is a candidate is partially negatively correlated.
Therefore, the probability that a vertex is matched after $k$ randomized rounds is strictly greater than $1 - 2^{-k}$ for any $k \ge 2$.
See Lemma~\ref{lem:OCS} for a precise statement.

\begin{algorithm}[t]
    \caption{Online Correlated Selection (OCS)}
    \label{alg:OCS}
    \begin{algorithmic}
        \medskip
        \STATEx \textbf{State variables:}
        \begin{itemize}
            \item $\type_i \in \big\{ \chosen, \notchosen, \nonadapt \big\}$ for each offline vertex $i \in L$; 
            initially, let $\type_i = \nonadapt$.
        \end{itemize}
        \STATEx \textbf{On receiving $2$ candidate offline vertices $i_1$ and $i_2$ (for an online vertex $j \in R$):}
        \begin{enumerate}
            \item With probability $\frac{1}{2}$, let it be an \emph{oblivious step}:
            \begin{enumerate}
                \item Draw $\ell, m \in \{1, 2\}$ uniformly at random.
                \item Let $\type_{i_{-m}} = \nonadapt$.
                \item If $m = \ell$, let $\type_{i_m} = \chosen$; otherwise, let $\type_{i_m} = \notchosen$.
                %
                %
            \end{enumerate}
            \item Otherwise (i.e., with probability $\frac{1}{2}$), let it be an \emph{adaptive step}:
            \begin{enumerate}
                \item Draw $m \in \{1, 2\}$ uniformly at random.
                \item If $\type_{i_m} = \chosen$, let $\ell = -m$; \newline
                if $\type_{i_m} = \notchosen$, let $\ell = m$; \newline
                if $\type_{i_m} = \nonadapt$, draw $\ell \in \{1, 2\}$ uniformly at random.
                \item Let $\type_{i_1} = \type_{i_2} = \nonadapt$.
            \end{enumerate}
            \item Return $i_\ell$.
        \end{enumerate}
        \smallskip
    \end{algorithmic}
\end{algorithm}

We now demonstrate how it works.
It maintains a state variable $\type_i$ for each offline vertex $i \in L$.
If the state $\type_i$ is equal to $\chosen$ or $\notchosen$, it reflects the matching status of $i$ \emph{the last time when $i$ is an candidate}, and indicates that the information can be used the next time when $i$ is a candidate.
If the state $\type_i$ is equal to $\nonadapt$, it means that the matching status of offline vertex $i$ cannot be used the next time when $i$ is a candidate.

For each randomized round in the two-choice greedy algorithm, where $i_1$ and $i_2$ are the candidates, the OCS picks one of the two candidates as follows.
First, it decides whether this is an \emph{oblivious step}, or an \emph{adaptive step} uniformly at random.

In an oblivious step, it uses a fresh random bit to determines its choice $i_\ell$, $\ell \in \{1, 2\}$, to be returned in this round.
Then, it draws $m \in \{1, 2\}$ uniformly at random and sets $\type_{i_m}$ to reflect its matching status in this round;
$\type_{i_{-m}}$ is set to $\nonadapt$.
That is, the OCS forwards the random bit in this round to subsequent rounds for only one of the two candidates, chosen uniformly at random.

In an adaptive step, on the other hand, the OCS seeks to use the previous matching status of the candidates to determine its choice of $i_\ell$.
First, it draws $m \in \{1, 2\}$ uniformly at random, and checks the state variable of $i_m$.
To achieve negative correlation, the OCS let the matching status of $i_m$ in this round to be the opposite of the state.
That is, if the state was $\chosen$, indicating that $i_m$ was matched the last time when it was a candidate, the OCS would choose $i_{-m}$ this time, and vice versa;
if the state variable was equal to $\nonadapt$, the OCS would use a fresh random bit to determine $i_\ell$.
In either case, reset the states of both $i_1$ and $i_2$ to be $\nonadapt$.

\begin{lemma}
    \label{lem:OCS}
    For any fixed sequence of pairs of candidates, any offline vertex $i$, and any $k \ge 0$, the OCS ensures that after being a candidate $k$ times, $i$ is matched with probability at least $1 - 2^{-k} \cdot f_k$, where $f_k$ is defined recursively as:
    \begin{equation}
        \label{eqn:OCS-recurrence}
        f_k = 
        \begin{cases}
            1 & k = 0, 1 \\
            f_{k-1} - \frac{1}{16} f_{k-2} & k \ge 2
        \end{cases}
    \end{equation}
\end{lemma}

\begin{proof}
    Let $P^1 = \{i_1^1, i_2^1\}, P^2 = \{i_1^2, i_2^2\}, \dots, P^T = \{i_1^T, i_2^T\}$ be the sequence of pairs of candidate offline vertices.
    We start with an interpretation of the OCS in the language of graphs. 
    First, consider a graph $G^{\textrm{\em ex-ante}} = (V, E^{\textrm{\em ex-ante}})$ as follows which we shall refer to as the \emph{ex-ante dependence graph}. 
    To make a distinction with the vertices and edges in the matching problem, we shall refer to the vertices and edges in the dependence graph as \emph{nodes} and \emph{arcs} respectively.
    
    Let there be a node for each pair of candidates;
    we will refer to them as $1 \le t \le T$, i.e.:
    \[
        V = \big\{ t : 1 \le t \le T \big\}
        ~.
    \]
    
    Further, for any fixed offline vertex $i$, let there be a directed arc from $t_1$ to $t_2$ for any two consecutive times in which $i$ is a candidate, i.e.:
    \[
        E^{\textrm{\em ex-ante}} = \big\{ (t_1, t_2) : t_1 < t_2; \exists i \in L \textrm{ s.t.\ } i \in P^{t_1}, i \in P^{t_2}, \textrm{ and } \forall t_1 < t < t_2, i \notin P^t \big\}
        ~.
    \]
    
    See Figure~\ref{fig:dependence-graphs-ex-ante} for an example.

    \begin{figure}[t]
        \centering
        \begin{subfigure}{\textwidth}
            \includegraphics[width=\textwidth]{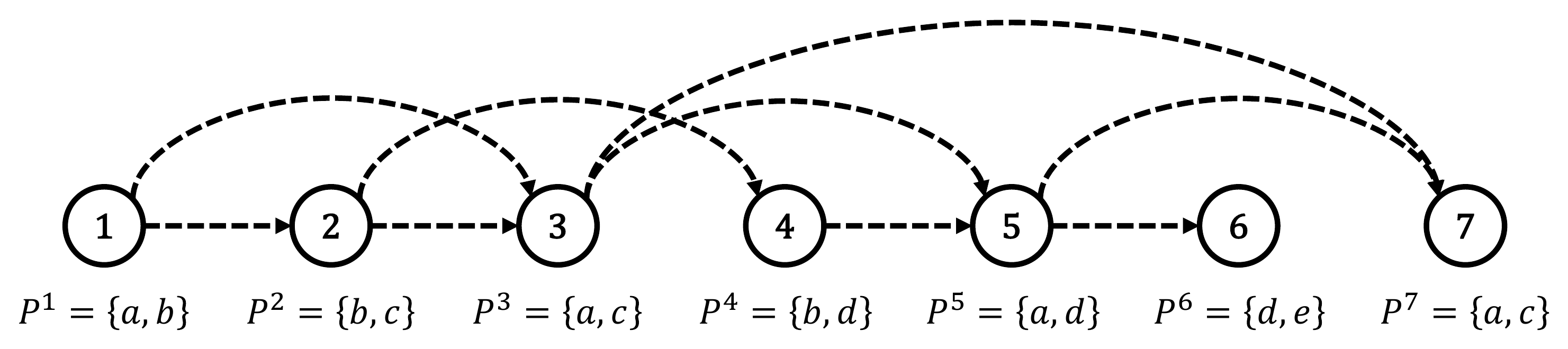}
            \caption{\emph{Ex-ante} dependence graph}
            \label{fig:dependence-graphs-ex-ante}
        \end{subfigure}
        
        \bigskip
        
        \begin{subfigure}{\textwidth}
            \includegraphics[width=\textwidth]{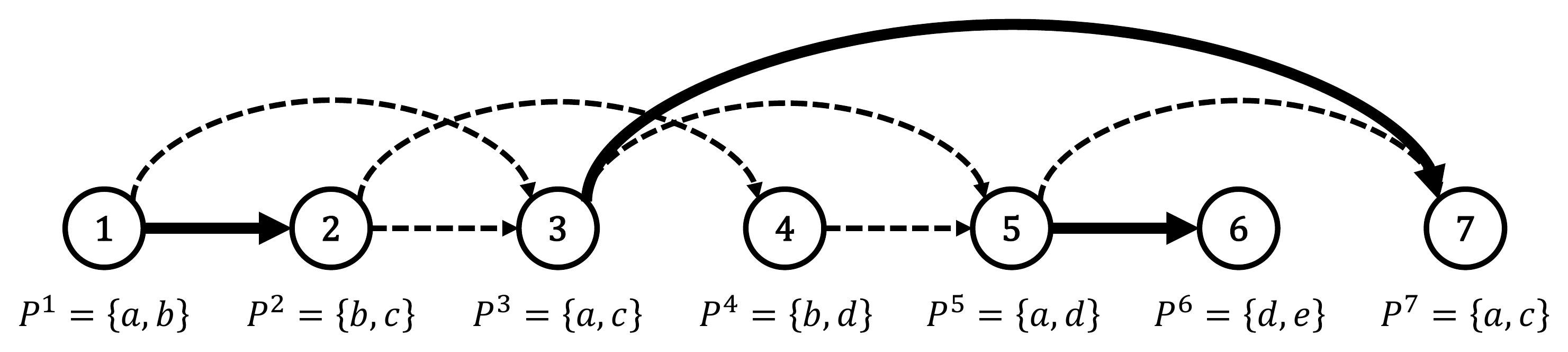}
            \caption{\emph{Ex-post} dependence graph (bold and solid edges)}
            \label{fig:dependence-graphs-ex-post}
        \end{subfigure}
        
        \bigskip
        
        \begin{subfigure}{\textwidth}
            \includegraphics[width=\textwidth]{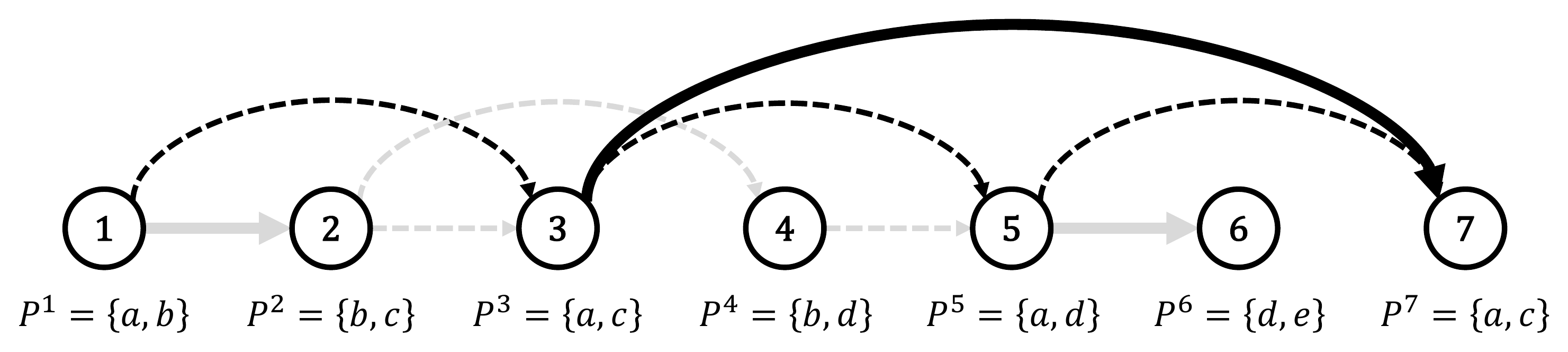}
            \caption{Dependence subgraph associated with a fixed candidate offline vertex, e.g., vertex $a$}
            \label{fig:dependence-graphs-fixed-candidate}
        \end{subfigure}
        \caption{An example with $5$ offline vertices and a sequence of $7$ pairs}
        \label{fig:dependence-graphs}
    \end{figure}

    Each arc in the \emph{ex-ante} dependence graph represents two steps in the sequence in which the OCS may use the same random bit to determine the offline vertices it returns.
    There are at most $2$ outgoing arcs and at most $2$ incoming arcs for each node by definition.
    
    In particular, consider any arc $(t_1, t_2)$ in the \emph{ex-ante} dependence graph, with $i$ being the common candidate.
    If the randomness used by the OCS is such that (1) step $t_1$ is an oblivious step, (2) $i_m = i$ in step $t_1$, (3) step $t_2$ is an adaptive step, and (4) $i_m = i$ in step $t_2$, the matching status of $i$ would be perfectly negatively correlated in the sense that $i$ is chosen in exactly one of the two steps.
    Each of the $4$ events happens independently with probability $\frac{1}{2}$. 
    
    The \emph{ex-post} dependence graph $G^{\textrm{\em ex-post}} = (V, E^{\textrm{\em ex-post}})$ is a subgraph of the \emph{ex-ante} dependence graph, keeping the arcs which correspond to the pairs of time steps that are perfectly negatively correlated, given the realization of whether each step is oblivious or adaptive, and the value of $m$ therein. 
    Equivalently, the \emph{ex-post} dependence graph is realized as follows.
    Over the randomness with which the OCS decides whether each step is oblivious or adaptive, and the value of $m$, each node in the \emph{ex-ante} dependence graph picks at most one of its incident arcs, each with probability $\frac{1}{4}$;
    an arc is realized in the \emph{ex-post} graph if both incident nodes pick it.
    With this interpretation, we get that the \emph{ex-post} graph is a matching.
    The OCS may be viewed as a randomized online algorithm that picks a matching of the \emph{ex-ante} graph, such that each arc in the \emph{ex-ante} graph is chosen with probability lower bounded by a constant.
    See Figure~\ref{fig:dependence-graphs-ex-post} for an example.
    
    Finally, we lower bound the probability that an offline vertex $i$ remains unmatched after being a candidate in $k$ steps in the sequence.
    Let $t_1 < t_2 < \dots$ be the time steps in which $i$ is a candidate.
    We will use offline vertex $a$ and $k = 4$ in  Figure~\ref{fig:dependence-graphs-fixed-candidate} as a running example, where $t_1 = 1, t_2 = 3, t_3 = 5, t_4 = 7$ and the relevant arcs in the dependence graphs are $(1, 3)$, $(3, 5)$, $(3, 7)$, and $(5, 7)$.
    
    If at least one of the arcs among $t_1 < t_2 < \dots < t_k$ are realized in the \emph{ex-post} dependence graph, vertex $i$ must be matched after step $t_k$.
    This is because the randomness (related to the choice of $\ell$ in the OCS) is perfectly negatively correlated in the two incident nodes of the arc and thus, $i$ is chosen exactly once in these two steps.
    For example, given that the arc $(3, 7)$ is realized in Figure~\ref{fig:dependence-graphs-fixed-candidate}, vertex $a$ must be matched after step $7$.
    
    On the other hand, if none of these arcs are realized, the random bits used in the $k$ steps $t_1 < t_2 < \dots < t_k$ are independent.
    For example, consider vertex $a$ and $k = 3$ in Figure~\ref{fig:dependence-graphs-fixed-candidate}; 
    vertex $a$ is chosen independently with probability $\frac{1}{2}$ in steps $t_1 = 1$, $t_2 = 3$, and $t_3 = 5$, given that neither $(1, 3)$ nor $(3, 5)$ is realized.
    
    Importantly, even if some of these steps are adaptive in that the matching decisions are based on the random bits realized earlier in some oblivious steps, from $i$'s viewpoint, they are still independent of the random bits in the other rounds that $i$ is involved in.
    For example, from $c$'s viewpoint in Figure~\ref{fig:dependence-graphs-fixed-candidate}, even though the matching decision in step $2$ is determined by that in step $1$, it is independent of the matching decisions in steps $3$ and $7$ that $c$ is involved in.
    
    Putting together, the probability that $i$ is \emph{unmatched} after steps $t_1 < t_2 < \dots < t_k$ is equal to (1) the probability that none of the arcs among these steps is realized, times (2) all $k$ independent random bits are against $i$.
    The latter is equal to $2^{-k}$.
    It remains to analyze the former; 
    we shall upper bound it by the probability that none of the arcs $(t_1, t_2), (t_2, t_3), \dots, (t_{k-1}, t_k)$ is realized.
    Denote this event as $F_k$ and its probability as $f_k$.
    
    Trivially, we have $f_0 = f_1 = 1$.
    To prove that the stated recurrence in Eqn.~\eqref{eqn:OCS-recurrence} governs $f_k$, we further divide the event $F_k$ into two subevents.
    Let $A_k$ the event that none of the arcs $(t_1, t_2), (t_2, t_3), \dots, (t_{k-1}, t_k)$ is realized, and step $t_k$ picks arc $(t_k, t_{k+1})$ in realizing the \emph{ex-post} dependence graph.
    Let $B_k$ be the event that none of the arcs is realized and step $t_k$ does not pick arc $(t_k, t_{k+1})$.
    Let $a_k$ and $b_k$ be the probability of $A_k$ and $B_k$ respectively.
    We have that $A_k$ and $B_k$ form a partition of $F_k$, and thus:
    \[
        f_k = a_k + b_k
        ~.
    \]
    
    If step $t_k$ picks arc $(t_k, t_{k+1})$, which happens with probability $\frac{1}{4}$, arc $(t_{k-1}, t_k)$ must not be realized by definition.
    Therefore, conditioned on the choice of $t_k$, $A_k$ happens if and only if the choices made by steps $t_1, t_2, \dots t_{k-1}$ is such that none of $(t_1, t_2), \dots, (t_{k-2}, t_{k-1})$ is realized, i.e., when $F_{k-1}$ happens.
    That is:
    \[
        a_k = \frac{1}{4} f_{k-1}
        ~.
    \]
    
    On the other hand, if step $t_k$ does not pick $(t_k, t_{k+1})$, there are two possibilities.
    The first case is when $t_k$ picks $(t_{k-1}, t_k)$, which happens with probability $\frac{1}{4}$.
    In this case, the choices made by $t_1, \dots, t_{k-1}$ must be such that none of $(t_1, t_2), \dots, (t_{k-2}, t_{k-1})$ is realized, and $t_{k-1}$ does \emph{not} pick $(t_{k-1}, t_k)$, i.e., $B_{k-1}$ happens.
    The second case is when $t_k$ picks neither $(t_{k-1}, t_k)$ nor $(t_k, t_{k+1})$, which happens with probability $\frac{1}{2}$.
    In this case, the choices made by $t_1, \dots, t_{k-1}$ must be such that none of $(t_1, t_2), \dots, (t_{k-2}, t_{k-1})$ is realized, i.e., $F_{k-1}$ happens.
    Putting together, we have:
    \[
        b_k = \frac{1}{4} b_{k-1} + \frac{1}{2} f_{k-1}
        ~.
    \]
    
    Eliminating $a_k$'s and $b_k$'s by combining the above three equations, we get the recurrence stated in Eqn.~\eqref{eqn:OCS-recurrence}.
\end{proof}

In fact, we can show a stronger version, which will be useful in the weighted case.
We say that $t_1 < t_2 < \dots < t_k$ is a consecutive sequence of rounds in which $i$ is involved if $i$ is a candidate in these rounds, but in no other rounds in between.

\begin{lemma}
    \label{lem:OCS-generalized}
    For any fixed sequence of pairs of candidates, any fixed candidate $i$, and any disjoint consecutive sequences of rounds of lengths $k_1, k_2, \dots, k_m \ge 1$ in which $i$ is involved, the OCS in Algorithm~\ref{alg:OCS} ensures that $i$ is chosen in at least one of the rounds with probability at least:
    \[
        1 - \prod_{\ell=1}^m 2^{-k_\ell} \cdot f_{k_\ell} 
        ~.
    \]
\end{lemma}

\begin{proof}
    Let $t_1^\ell < t_2^\ell < \dots, t_{k_\ell}^\ell$ be the $\ell$-th consecutive sequence of round in which $i$ is involved, for any $1 \le \ell \le m$. 
    The probability that $i$ is never chosen is equal to (1) the probability that none of the arcs among the steps in these sequences is realized, times (2) the probability that all $\sum_{\ell=1}^m k_\ell$ random bits are against $i$.
    The latter is $\prod_{\ell=1}^m 2^{-k_\ell}$.
    We upper bound the former with the probability that for any $1 \le \ell \le m$, none of the arcs $(t_1^\ell, t_2^\ell), \dots, (t_{k_\ell-1}^\ell t_{k_\ell}^\ell)$ is realized.
    Further, note that the events are independent for different $\ell$, as each event only relies on the choice made by the nodes in the corresponding subsequence.
    Hence, it is at most $\prod_{\ell=1}^m f_{k_\ell}$.
\end{proof}

\subsection{Beating $\frac{1}{2}$ Using the OCS}

Finally, we show that the two-choice greedy algorithm is strictly better than $\frac{1}{2}$-competitive when it is combined with the OCS to have partial negative correlation in the randomized rounds.

\begin{theorem}
    \label{thm:main-unweighted}
    Two-choice greedy, with the randomized rounds implemented using the OCS, is at least $\Gamma = 0.505$-competitive.
\end{theorem}

\begin{proof}
    We will maintain $x_i = 1 - 2^{-k_i} \cdot f_{k_i}$ for each offline vertex $i$ as a lower bound of the probability that $i$ is matched.
    For each online vertex $j$, let $x_{ij}$ be the increment in $x_i$ due to $j$.
    Then, we have:
    \[
        \primal = \sum_{i \in L} x_i
        ~.
    \]
    
    Recall that $\alg$ denote the expected objective given by the algorithm.
    The choices of $B(j)$'s by the two-choice greedy algorithm is independent of the random bits used by the OCS.
    Therefore, the sequence of pairs of candidates in the randomized rounds, which are sent to the OCS, is fixed.
    By Lemma~\ref{lem:OCS}, we have:
    \[
        \alg \ge \primal
        ~.
    \]
    
    
    To prove the stated competitive ratio, it remains to explain how to maintain a dual assignment such that (1) the dual objective is no more than the primal one, i.e., $\dual \le \primal$, and (2) it is approximately feasible up to a $\Gamma$ factor, i.e., $\alpha_i + \beta_j \ge \Gamma$ for every edge $(i, j)$.
    
    \subsubsection*{Dual Updates}
    
    It is based on the solution to the LP in the following lemma, whose proof is deferred to the end of the subsection.
    
    \begin{lemma}
        \label{lem:ratio-lp}
        The optimal value of the LP below is at least $0.505$:
        \begin{align}
            \text{\rm maximize} \quad & \Gamma \notag \\[1ex]
            \text{\rm subject to} \quad & \Delta \alpha(k) + \frac{1}{2} \beta(k) \le 2^{-k} \cdot f_k - 2^{-k-1} \cdot f_{k+1} & \forall k \ge 0 \label{eqn:lp-gain-split} \\
            & \sum_{\ell = 0}^{k-1} \Delta\alpha(\ell) + \beta(k) \ge \Gamma & \forall k \ge 0 \label{eqn:lp-dual-feasibility} \\
            & \beta(k) \ge \beta(k+1) & \forall k \ge 0 \label{eqn:lp-monotone-beta} \\[2ex]
            & \Delta\alpha(k), \beta(k) \ge 0 & \forall k \ge 0 \notag
        \end{align}

    \end{lemma}
    
    Consider an online vertex $j \in R$.
    Let $\kmin = \min_{i \in N(j)} k_i$ be the minimum number of randomized rounds in which a fixed neighbor of $j$ is a candidate.
    
    Suppose it is a randomized round. 
    Then, we have $k_i = \kmin$ for both $i \in B(j)$.
    For both $i \in B(j)$, $x_i$ increases by $2^{-\kmin} \cdot f_{\kmin} - 2^{-\kmin-1} \cdot f_{\kmin+1}$. 
    In the dual, increase $\alpha_i$ by $\Delta \alpha(\kmin)$ for both $i \in B(j)$, and let $\beta_j = \beta(\kmin)$. 
    
    Suppose it is a deterministic round.
    Let $i$ be the only vertex in $B(j)$. 
    Then, $x_i$ increases by $2^{-\kmin} \cdot f_{\kmin}$.
    In the dual, increase $\alpha_i$ by $\sum_{\ell \ge \kmin} \Delta \alpha(\ell)$, and let $\beta_j = \beta(\kmin+1)$.
    
    No update is needed in an unmatched round, as $\primal$ remains the same. 
    
    \subsubsection*{Objective Comparisons}
    
    Next, we show that the increment in the dual objective $\dual$ is at most that in the primal objective $\primal$.
    In a randomized round, it follows by Eqn.~\eqref{eqn:lp-gain-split}.
    In a deterministic round, it follows by a sequence of inequalities below:
    \begin{align*}
        \sum_{\ell \ge k_i} \Delta \alpha(\ell) + \beta(k_i+1) 
        & 
        \le \sum_{\ell \ge k_i} \Delta \alpha(\ell) + \frac{1}{2} \beta(k_i) + \frac{1}{2} \beta(k_i+1) && \textrm{(Eqn.~\eqref{eqn:lp-monotone-beta})} \\
        & 
        \le \sum_{\ell \ge k_i} \big( \Delta \alpha(\ell) + \frac{1}{2} \beta(\ell) \big) \\
        & 
        \le \sum_{\ell \ge k_i} \big( 2^{-\ell} \cdot f_\ell - 2^{-\ell-1} \cdot f_{\ell+1} \big) && \textrm{(Eqn.~\eqref{eqn:lp-gain-split})} \\[1ex]
        &
        = 2^{-k_i} \cdot f_{k_i}
        ~.
    \end{align*}

    \subsubsection*{Approximate Dual Feasibility.}
    
    We first summarize the following invariants which follow by the definition of the dual updates.
    \begin{itemize}
        \item For any offline vertex $i \in L$, $\alpha_i = \sum_{\ell=0}^{k_i-1} \Delta \alpha(\ell)$.
        \item For any online vertex $j$, $\beta_j = \beta(k)$ if it is matched either in a randomized round to neighbors with $k_i = k$, or in a deterministic round to a neighbor with $k_i = k-1$.
    \end{itemize}
    
    \clearpage 
    
    For any edge $(i, j) \in E$, consider the value of $k_i$ at the time when $j$ arrives.
    If $k_i = \infty$, the value of $\alpha_i$ alone ensures approximately dual feasibility because:
    \begin{align*}
        \alpha_i 
        & 
        = \sum_{\ell \ge 0} \Delta \alpha(\ell) \\
        &
        = \lim_{k \to \infty} \sum_{\ell = 0}^{k-1} \Delta \alpha(\ell) \\[1ex]
        &
        \ge \Gamma - \lim_{k \to \infty} \beta(k) && \textrm{(Eqn.~\eqref{eqn:lp-dual-feasibility})} \\[2ex]
        &
        = \Gamma 
        ~.
        && \textrm{(Eqn.~\eqref{eqn:lp-gain-split}, whose RHS tends to $0$)}
    \end{align*}
    
    Otherwise, by the definition of the two-choice greedy algorithm, $j$ is matched either in a randomized round to two vertices with $k_{i'} \le k_i$, or in a deterministic round to a vertex with $k_{i'} < k_i$.
    In both cases, we have:
    \[
        \beta_j \ge \beta(k_i)
        ~.
    \]
    
    Approximate dual feasibility now follows by $\alpha_i = \sum_{\ell = 0}^{k_i-1} \Delta \alpha(\ell)$ and Eqn.~\eqref{eqn:lp-dual-feasibility}.
\end{proof}

\begin{proof}[Proof of Lemma~\ref{lem:ratio-lp}]
    Consider a restricted version of the LP which is finite.
    For some positive $\kmax$, let $\Delta \alpha(k) = \beta(k) = 0$ for all $k > \kmax$.
    Then, it becomes:
    \begin{align*}
        \text{\rm maximize} \quad & \Gamma \notag \\[2ex]
        \text{\rm subject to} \quad & \Delta \alpha(k) + \frac{1}{2} \beta(k) \le 2^{-k} \cdot f_k - 2^{-k-1} \cdot f_{k+1} & 0 \le k \le \kmax \\[1ex]
        & \sum_{\ell = 0}^{k-1} \Delta\alpha(\ell) + \beta(k) \ge \Gamma & 0 \le k \le \kmax \\
        & \sum_{\ell = 0}^{\kmax} \Delta \alpha(\ell) \ge \Gamma \\[1.5ex]
        & \beta(k) \ge \beta(k+1) & 0 \le k < \kmax \\[4ex]
        & \Delta\alpha(k), \beta(k) \ge 0 & 0 \le k \le \kmax \notag
    \end{align*}
 
    See Table~\ref{tab:lp-solution} for an approximately optimal solution for the relaxed LP with $\kmax = 7$.
    It gives an objective $\Gamma \approx 0.5051$.
\end{proof}

\begin{table}[t]
    \centering
    \begin{tabular}{c|ccc}
        \hline
        $k$     & $f_k$     & $\Delta \alpha(k)$    & $\beta(k)$    \\
        \hline
        $0$     & $1.00000000$     & $0.24744702$             & $0.50510596$     \\
        $1$     & $1.00000000$     & $0.13679553$             & $0.25765895$     \\
        $2$     & $0.93750000$  & $0.06456829$             & $0.12086342$     \\
        $3$     & $0.87500000$   & $0.03020205$             & $0.05629513$     \\
        $4$     & $0.81640625$ & $0.01417514$             & $0.02609308$     \\
        $5$     & $0.76171875$ & $0.00674016$             & $0.01191794$     \\
        $6$     & $0.71069336$ & $0.00333533$             & $0.00517779$     \\
        $7$     & $0.66308594$ & $0.00184246$             & $0.00184246$     \\
        \hline
    \end{tabular}
    \caption{An approximately optimal solution to the restricted LP for $\kmax = 7$, rounded to the $8$-th digit after the decimal point, whose $\Gamma \approx 0.5051$.}
    \label{tab:lp-solution}
\end{table}

\bibliographystyle{plainnat}
\bibliography{matching.bib}

\end{document}